\documentclass[a4paper,UKenglish,cleveref,autoref]{lipics-v2019}
\nolinenumbers
\usepackage{etoolbox}
\apptocmd{\thebibliography}{\raggedright}{}{}
\usepackage{mathtools,amssymb,mathrsfs}
\usepackage{yfonts}
\newcommand{\eps}{\epsilon}
\newcommand{\astralbody}{view}
\newcommand\dimcl[3]{(#1,#2,#3)}

\newcommand\pnt[1]{\mathbf{#1}}
\DeclarePairedDelimiter{\ceil}{\lceil}{\rceil}

\renewcommand{\O}{\mathcal{O}}

\newcommand\smallO{
  \mathchoice
    {{\scriptstyle\mathcal{O}}}
    {{\scriptstyle\mathcal{O}}}
    {{\scriptscriptstyle\mathcal{O}}}
    {\scalebox{.7}{$\scriptscriptstyle\mathcal{O}$}}
}

\title{Path and Ancestor Queries over Trees with Multidimensional Weight Vectors}
\titlerunning{Path and Ancestor Queries over Trees}

\author{Meng He}{Faculty of Computer Science, Dalhousie University, Canada}{mhe@cs.dal.ca}{}{}
\author{Serikzhan Kazi}{Faculty of Computer Science, Dalhousie University, Canada}{skazi@dal.ca}{}{}

    \authorrunning{M., He and S., Kazi}
    \Copyright{Meng He and Serikzhan Kazi}

\funding{{This work was supported by NSERC of Canada.}}

\begin{document}
\maketitle 
\begin{abstract}
    We consider an ordinal tree $T$ on $n$ nodes, with each node assigned a 
    $d$-dimensional weight vector $\pnt{w} \in \{1,2,\ldots,n\}^d,$ where $d \in \mathbb{N}$ is a constant.
	We study path queries as generalizations of well-known {\textit{orthogonal range queries}},
	with one of the dimensions being tree topology rather than a linear order.
    Since in our definitions $d$ only represents the number of dimensions of
    the weight vector without taking the tree topology into account,
    a path query in a tree with $d$-dimensional weight vectors generalize
    the corresponding $(d+1)$-dimensional orthogonal range query.
	We solve {\textit{ancestor dominance reporting}} problem as a direct generalization
	of dominance reporting problem, 
    in time $\O(\lg^{d-1}{n}+k)$ 
    and space of $\O(n\lg^{d-2}n)$ words,
	where $k$ is the size of the output, for $d \geq 2.$
    We also achieve a tradeoff of $\O(n\lg^{d-2+\eps}{n})$ words of space, with
    query time of $\O((\lg^{d-1} n)/(\lg\lg n)^{d-2}+k),$ for the same problem, when $d \geq 3.$
    We solve {\textit{path successor problem}}
	in $\O(n\lg^{d-1}{n})$
	words of space and time
	$\O(\lg^{d-1+\eps}{n})$ for $d \geq 1$ and an arbitrary constant $\eps > 0.$ 
 	We propose a solution to {\textit{path counting problem}},
	with $\O(n(\lg{n}/\lg\lg{n})^{d-1})$ words of space and $\O((\lg{n}/\lg\lg{n})^{d})$ query time, 
    for $d \geq 1.$
 	Finally, we solve {\textit{path reporting problem}} in $\O(n\lg^{d-1+\eps}{n})$ words of space
    and $\O((\lg^{d-1}{n})/(\lg\lg{n})^{d-2}+k)$ query time, for $d \geq 2.$
    These results match or nearly match the best tradeoffs of the respective range queries.
    We are also the first to solve path successor even for $d = 1$. 
\end{abstract}


\section {Introduction}\label{section:introduction}

    The problem of preprocessing a weighted
	tree, i.e., a tree in which each node is associated with a weight value, to support various queries 
	evaluating a certain function on the node weights
	of a given path, has been extensively studied
	~\cite{Alon87optimalpreprocessing,Chazelle1987,Hagerup:2000:PPP:333896.333905,DBLP:journals/njc/KrizancMS05,DBLP:journals/algorithmica/DemaineLW14,He:2016:DSP:2983296.2905368,DBLP:journals/algorithmica/ChanHMZ17}.
    For example, in {\textit{path counting}} (resp. {\textit{path reporting}}),
	the nodes of the given path with weights lying in the given query interval are counted (resp. reported).
    These queries address the needs of fast information retrieval from tree-structured data such as XML and tree network topology. 

    For many applications, meanwhile, a node in a tree is associated with not just a single weight, but rather with a vector of weights.
    Consider a simple scenario of an online forum thread,
    where users can rate responses and respond to posts.
    Induced is 
    a tree-shaped structure with posts representing nodes,
    and replies to a post being its children.
    One can imagine enumerating all the ancestor posts
    of a given post that are not too short and
    have sufficiently high average ratings. Ancestor dominance query,
    which is among the problems we consider, provides an appropriate
    model in this case.

    We define a $d$-{\textit{dimensional weight vector}} $\pnt{w} = (w_1,w_2,\ldots,w_d)$
    to be a vector with $d$ components, each in rank space $[n],$
    \footnote{Throughout this paper, $[n]$ stands
    for $\{1,2,\ldots,n\}$ for any positive integer $n$}
    i.e. $\pnt{w} \in [n]^d,$
    with $w_i$ being referred to as the $i$th weight of $\pnt{w}.$
    We then consider an ordinal tree $T$ on $n$ nodes, 
    each node $x$ of which is assigned a $d$-dimensional weight vector $\pnt{w}(x).$  
    The queries we will define all give a $d$-dimensional orthogonal range 
    $Q = \prod_{i=1}^{d}[q_i,q'_i],$
    and a weight vector $\pnt{w}$ is in $Q$ {\textit{iff}} for any $i \in [1,d],\,q_i \leq w_i \leq q'_i$ holds.
	In our queries, then,
    we are given a pair of vertices $x,y \in T,$ and an arbitrary orthogonal range $Q.$
    With $P_{x,y}$ being the path from $x$ to $y$ in the tree $T,$
	the goal is to preprocess the tree $T$ for the following
	types of queries:
	\begin{itemize}
        \item {\emph{Path Counting}}: return $|\{z \in P_{x,y}\,|\,\pnt{w}(z) \in {Q}\}|$. 
        \item {\emph{Path Reporting}}: enumerate $\{z \in P_{x,y}\,|\,\pnt{w}(z) \in {Q}\}$.
        \item {\emph{Path Successor}}: return $\mathtt{argmin}\{w_1(z)\,|\,z \in P_{x,y}\,\text{and}\,\pnt{w}(z) \in Q\}$.\footnote{For path successor, 
    we assume that $q'_1 = \infty$;
    if not, we need only check whether the $1$st weight of the returned node is at most $q'_1$.}
        \item {\emph{Ancestor Dominance Reporting}}: a special case of path reporting, in which $y$ is the root of the tree and $q_i' = +\infty$ for all $i \in [d].$
            That is, the query reports the ancestors of $x$ whose weight vectors dominate the vector $\pnt{q} = (q_1, q_2, \ldots, q_d)$. 
	\end{itemize}
	This indeed is a natural generalization of the traditional
    weighted tree, which we refer to as ``scalarly-weighted'',
    to the case when the weights 
    are multidimensional vectors. At the same time,
    when the tree degenerates into a single path, 
    these queries become respectively $(d+1)$-dimensional orthogonal range counting, reporting, successor, 
	as well as $(d+1)$-dimensional dominance reporting queries.
    Thus, the queries we study are generalizations of these fundamental geometric queries in high dimensions.
    We also go along with the state-or-art in orthogonal range search by considering weights in rank space,
    since the case in which weights are from a larger universe can be reduced to it~\cite{Gabow:1984:SRT:800057.808675}.

	\subsection{Previous Work}
    \subparagraph*{Path Queries in Weighted Trees.}
    For scalarly-weighted trees, Chazelle~\cite{Chazelle1987}
    gave an $\O(n)$-word {\textit{emulation dag}}-based data structure
    that answers path counting queries in $\O(\lg{n})$ time;\footnote{$\lg{x}$ denotes $\log_2{x}$ in this paper.}
    it works primarily with topology of the tree and is thus oblivious to the distribution of weights.
	Later, He et al.~\cite{He:2016:DSP:2983296.2905368} proposed
    a solution with $nH(W_T)+\smallO(n\lg\sigma)$
    bits of space and $\O(\frac{\lg{\sigma}}{\lg\lg{n}}+1)$ query time,
    when the weights are from $[\sigma];$
    here, $H(W_T)$ is the entropy of the multiset of the weights in $T.$
    When $\sigma \ll n,$ this matters.

    He et al.~\cite{He:2016:DSP:2983296.2905368} introduced and solved path reporting problem
	using linear space and $\O((1+k)\lg\sigma)$ query time, and $\O(n\lg\lg\sigma)$ words of space
	but $\O(\lg\sigma+k\lg\lg\sigma)$ query time, in the word-RAM model;
    henceforth we reserve $k$ for the size of the output.
	Patil et al.~\cite{DBLP:journals/jda/PatilST12}
	presented a succinct data structure for path reporting 
    with $n\lg\sigma+6n+\smallO(n\lg\sigma)$ bits of space
    and $\O((\lg{n}+k)\lg\sigma)$ query time. 
	An optimal-space solution with $nH(W_T)+\smallO(n\lg\sigma)$ bits of space
    and $\O((1+k)(\frac{\lg\sigma}{\lg\log{n}}+1))$ reporting time
    is due to He et al.~\cite{He:2016:DSP:2983296.2905368}.
    One of the tradeoffs proposed by Chan et al.~\cite{DBLP:journals/algorithmica/ChanHMZ17},
    requires $\O(n\lg^{\eps}{n})$ words of space 
    for the query time of $\O(\lg\lg{n}+k).$

    \subparagraph*{Orthogonal Range Queries.}
	Dominance reporting in $3D$ was solved by
    Chazelle and Edelsbrunner~\cite{DBLP:journals/dcg/ChazelleE87} in linear
	space with either $\O((1+k)\lg{n})$ or $\O(\lg^{2}n+k)$ time,
    in pointer-machine (PM) model, with the latter being improved to
    $\O(\lg{n}\lg\lg{n}+k)$ by Makris and Tsakalidis~\cite{DBLP:journals/ipl/MakrisT98}.
    Same authors~\cite{DBLP:journals/ipl/MakrisT98} developed,
    in the word-RAM, a linear-size, $\O(\log{n}+k)$ and $\O((\lg\lg{n}\lg\lg\lg{n}+k)\lg\lg{n})$ query-time data
    structures for the unrestricted case and for points in rank space, respectively.
    Nekrich~\cite{DBLP:conf/compgeom/Nekrich07} presented a word-RAM data structure
    for points in rank space,
	supporting queries in $\O((\lg\lg{n})^{2}+k)$ time, and occupying $\O(n\lg{n})$ words;
    this space was later reduced to linear by Afshani~\cite{DBLP:conf/esa/Afshani08},
    retaining the same query time. Finally, 
    in the same model, a linear-space solution with $\O(\log\log{n}+k)$ query time
    was designed for $3D$ dominance reporting in rank space~\cite{DBLP:conf/esa/Afshani08,Chan:2011:PPS:2133036.2133121}.
    In the PM model, Afshani~\cite{DBLP:conf/esa/Afshani08}
    also presented an $\O(\log{n}+k)$ query time, linear-space data structure
    for the points in $\mathbb{R}^{3}.$

	For the word-RAM model, 
    J{\'{a}}J{\'{a}} et al.~\cite{DBLP:conf/isaac/JaJaMS04} 
    generalized the range counting problem for $d \geq 2$ dimensions
    and proposed a data structure with
	$\O(n(\lg{n}/\lg\lg{n})^{d-2})$ words of space
	and $\O((\lg{n}/\lg\lg{n})^{d-1})$ query time.
    Chan et al.~\cite{DBLP:conf/compgeom/ChanLP11} solved orthogonal range reporting in 
    $3D$ rank space in $\O(n\lg^{1+\eps}{n})$ words of space and $\O(\lg\lg{n}+k)$ query time.

    Nekrich and Navarro~\cite{DBLP:conf/swat/NekrichN12} proposed two tradeoffs for range successor, 
	with either $\O(n)$ or $\O(n\lg\lg{n})$ words of space, and respectively with $\O(\lg^{\eps}{n})$ or $\O((\lg\lg{n})^2)$ query time. 
	Zhou~\cite{DBLP:journals/ipl/Zhou16} later improved upon the query time of the second tradeoff by a factor of $\lg\lg{n},$ within the same space.
    Both results are for points in rank space.

    \subsection{Our Results}
    As $d$-dimensional path queries generalize the corresponding $(d+1)$-dimensional orthogonal range queries,
    we compare results on them to show that our bounds match or nearly match the best results or some of the best tradeoffs on geometric queries in Euclidean space. 
    We present solutions for the (we assume $d$ is a positive integer constant):
    \begin{itemize}
        \item ancestor dominance reporting problem, in $\O(n\lg^{d-2}{n})$ words of space and
              $\O(\lg^{d-1}n+k)$ query time for $d \geq 2$.
              When $d=2,$ this matches the space bound for $3D$ dominance reporting of~\cite{DBLP:conf/esa/Afshani08,Chan:2011:PPS:2133036.2133121},
              while still providing efficient query support. 
              When $d \geq 3,$ we also achieve a tradeoff of $\O(n\lg^{d-2+\eps}{n})$ words of space, with
              query time of $\O(\lg^{d-1} n/(\lg\lg n)^{d-2}+k);$
          \item path successor problem,
            in $\O(n\lg^{d-1}{n})$ 
            words and $\O(\log^{d-1+\eps}{n})$
            query time, for an arbitrarily small positive constant $\eps,$ and $d \geq 2$. 
            These bounds match the first tradeoff for range successor of Nekrich and Navarro~\cite{DBLP:conf/swat/NekrichN12}.
            \footnote{which can be generalized to higher dimensions via standard techniques based on range trees}
            Previously this problem has not been studied even on scalarly-weighted trees;
        \item path counting problem,
            in $\O(n(\frac{\log{n}}{\log\log{n}})^{d-1})$ words of space
            and $\O((\frac{\log{n}}{\log\log{n}})^{d})$ query time for $d \geq 1.$
            This matches the best bound for range counting in $d+1$ dimensions~\cite{DBLP:conf/isaac/JaJaMS04};
        \item path reporting problem, 
            in $\O(n\lg^{d-1+\eps}{n})$ words of space
            and $\O((\lg^{d-1}{n})/(\lg\lg{n})^{d-2}+k)$ query time, for $d \geq 2.$
            When $d=2$, the space matches that of the corresponding result of Chan et al.~\cite{DBLP:conf/compgeom/ChanLP11} on $3D$ range reporting, 
            while the first term in the query complexity is slowed down by a sub-logarithmic factor.
    \end{itemize}
    To achieve our results,
    we introduce a framework for solving range sum queries in arbitrary semigroups
    and extend base-case data structures
    to higher dimensions using universe reduction. A careful
    design with results hailing from succinct data structures
    and tree representations has been necessary,
    both for building space- and time-efficient base data structures,
    and for porting, using {\textit{tree extractions}}, 
    the framework of range trees decompositions
    from general point-sets to tree topologies (\Cref{lem:reductionBinaryCase}).
    We employ a few novel techniques, such as extending the notion of {\textit{maximality}} in Euclidean
    sense to tree topologies, and providing the means of efficient computation thereof (\Cref{section:ancestralReportingSection}).
    Given a weighted tree $T,$ 
    we propose efficient means of zooming into the nodes of $T$
    with weights in the given range in the range tree (\Cref{lemma:markedView}).
    Given the ubiquitousness of the concepts, these technical contributions are likely to be of independent interest.

\section {Preliminaries}\label{section:preliminaries}
\subparagraph*{Notation.}
Given a $d$-dimensional weight vector $\pnt{w} = (w_1,w_2,\ldots,w_d)$, we define vector $\pnt{w}_{i,j}$ to be $(w_i,w_{i+1},\ldots,w_j)$. 
We extend the definition to a range $Q = \prod_{i=1}^{d}[q_i,q'_i]$ by setting $Q_{i,j} = \prod_{k=i}^{j}[q_k,q'_k].$
We use the symbol $\succeq$ for domination:
$\pnt{p} \succeq \pnt{q}$ {\textit{iff}} $\pnt{p}$ dominates $\pnt{q}.$
With $d' \leq d$ and $0 < \eps < 1$ being constants,
a weight vector $\pnt{w}$ is said to be {\emph{$\dimcl{d'}{d}{\eps}-dimensional$}} {\textit{iff}} $\pnt{w} \in [n]^{d'}\times{}[\ceil{\lg^{\eps}{n}}]^{d-d'};$
i.e., each of its first $d'$ weights is drawn from $[n],$
while each of its last $d-d'$ weights is in $[\ceil{\lg^{\eps}{n}}]$. When stating theorems, we define $i/0 = \infty$ for $i > 0.$

During a preorder traversal of a given tree $T$, the $i$th node visited is said to have {\textit{preorder rank}} $i$. 
Preorder ranks are commonly used to identify tree nodes in various succinct data structures which we use as building blocks. 
Thus, we also identify a node by its preorder rank, i.e., node $i$ in $T$ is the node with preorder rank $i$ in $T$.
The path between the nodes $x,y \in T$ is denoted as $P_{x,y},$ both ends inclusive.
For a node $x \in T,$ its set of ancestors, denoted as $\mathcal{A}(x),$ includes $x$ itself;
$\mathcal{A}(x)\setminus{}\{x\}$ is then the set of {\textit{proper ancestors}} of $x.$
Given two nodes $x,y \in T,$ where $y \in \mathcal{A}(x)$, we set $A_{x,y} \triangleq P_{x,y}\setminus{}\{y\}.$

\subparagraph*{Succinct Representations of Ordinal Trees.}
Succinct representations of unlabeled and labeled ordinal trees is a widely researched area.
In a labeled tree, each node is associated with a label over an alphabet. Such a label can serve as a scalar weight;
in our solutions, however, they typically categorize tree nodes into different classes.
Hence we call these assigned values labels instead of weights.
We summarize the previous result used in our solutions, in which a node (resp. ancestor) with
label $\alpha$ is called an {\textit{$\alpha$-node}} (resp. {\textit{$\alpha$-ancestor}}):

\begin{lemma}[\cite{He:2016:DSP:2983296.2905368,DBLP:journals/talg/HeMS12}]\label{lemma:smallAlphabetTrees}
	Let $T$ be an ordinal tree on $n$ nodes, each having a label drawn from $[\sigma],$
 	where $\sigma = \O(\lg^{\eps}{n})$ for some constant $0 < \eps < 1.$
	Then, $T$ can be represented in $n(\lg{\sigma}+2)+\smallO(n)$ bits of space to support the
 	following operations, for any node $x \in T,$ in $\O(1)$ time: 
        {$\mathtt{child}(T,x,i)$}, the $i$-th child of $x$;
        {$\mathtt{depth}(T,x)$}, the number of ancestors of $x$;
        {$\mathtt{level\_anc}(T,x,i)$}, the $i$-th lowest proper ancestor of $x$;
        {$\mathtt{pre\_rank_{\alpha}}(T,x)$}, the number of $\alpha$-nodes that precede $x$ in preorder;
        {$\mathtt{pre\_select_{\alpha}}(T,i)$}, the $i$-th $\alpha$-node in preorder; and
        {$\mathtt{level\_anc_{\alpha}}(T,x,i)$}, the $i$-th lowest $\alpha$-ancestor of $x.$
\end{lemma}
\Cref{lemma:smallAlphabetTrees} also includes a result on representing an unlabeled ordinal tree,
which corresponds to $\sigma \equiv 1,$ in $2n+\smallO(n)$ bits~\cite{DBLP:journals/talg/HeMS12}. 
Another important special case is that of $\sigma = 2;$ here, $T$ is referred to as a {\textit{$0/1$-labeled tree}},
and the storage space becomes $3n+\smallO(n)$ bits. 

\subparagraph*{Tree Extraction.}
{\textit{Tree extraction}}~\cite{He:2016:DSP:2983296.2905368} filters out a subset of nodes
 while preserving the underlying ancestor-descendant relationship among the nodes.
 Namely, given a subset $X$ of tree nodes called {\textit{extracted nodes}}, 
 an {\textit{extracted tree}} $T_X$ can be obtained from the original tree $T$ as follows.
 Consider each node $v \notin X$ in an arbitrary order; let $p$ be $v$'s parent.
 We remove $v$ and all its incident edges from $T$, 
 and plug all its children $v_1,v_2,\ldots,v_k$ (preserving their left-to-right order) into the slot now freed from $v$ in $p$'s list of children.
 After removing all the non-extracted nodes, if the resulting forest $F_X$ is a tree, then $T_X \equiv F_X.$ 
 Otherwise, we create a dummy root $r$ and insert the roots of the trees in $F_X$ as the children of $r$, in the original left-to-right order.
 The preorder ranks and depths of $r$ are both $0$, so that those of non-dummy nodes still start at $1$.
 An original node $x \in X$ of $T$ and its copy, $x'$, in $T_X$ are said to {\textit{correspond}} to each other;
 $x'$ is also said to be the {\textit{$T_X$-{\astralbody}}} of $x,$ and $x$ is the {\textit{T-source}} of $x'.$
 The $T_X$-{\astralbody} of a node $y \in T$ ($y$ is not required to be in $X$) is more generally defined to be the node 
 $y' \in T_X$ corresponding to the lowest extracted ancestor of $y,$ i.e.
 to the lowest node in $\mathcal{A}(y) \cap{} X.$

\subparagraph*{Representation of a Range Tree on Node Weights by Hierarchical Tree Extraction.}
Range trees are widely used in solutions to query problems in Euclidean space. 
He et al.~\cite{He:2016:DSP:2983296.2905368} further applied the idea of range trees to scalarly-weighted trees. 
They defined a conceptual range tree on node weights and represented it by a hierarchy of tree extractions. 
We summarize its workings when the weights are in rank space.

We first define a conceptual range tree on $[n]$ with branching factor $f$, where $f = \O(\lg^{\eps} n)$ for some constant $0 < \eps < 1.$ 
Its root represents the entire range $[n]$.
Starting from the root level, we keep partitioning each range, $[a,b]$, 
at the current lowest level into $f$ child ranges $[a_1,b_1],\ldots,[a_f,b_f],$
where $a_i = \ceil{(i-1)(b-a+1)/f} + a$ and $b_i = \ceil{i(b-a+1)/f}+a-1$.
This ensures that, if weight $j \in [a,b]$, then $j$ is contained in the child range with subscript 
$\ceil{f(j-a+1)/(b-a+1)}$, which can be determined in $\O(1)$ time.
We stop partitioning a range when its size is $1$. 
This range tree has $h = \ceil{\log_f{n}} + 1$ levels. The root is at level $1$ and the bottom level is level $h$. 

For 
$1 \leq l < h,$ we construct an auxiliary tree $T_l$ for level $l$ of this range tree as follows: 
Let $[a_1,b_1],\ldots,[a_m,b_m]$ be the ranges at level $l$. 
For a range $[a,b],$ let $F_{a,b}$ stand for the extracted 
forest of the nodes of $T$ with weights in $[a,b].$ 
Then, for each range $[a_i,b_i],$ we extract $F_{a_i,b_i}$ 
and plug its roots as children of a dummy root $r_l$, retaining 
the original left-to-right order of the roots within the forest.
Between forests, the roots in $F_{a_{i+1},b_{i+1}}$ are the right siblings of the roots in  $F_{a_i,b_i}$, for any $i \in [m-1]$.
We then label the nodes of $T_l$ using the reduced alphabet $[f],$ 
as follows. Note that barring the dummy root $r_l$, 
there is a bijection between the nodes of $T$ and those of $T_l$. 
Let $x_l \in T_l$ be the node corresponding to $x \in T.$
In the range tree, let $[a,b]$ be the level-$l$ range containing the weight of $x.$
Then, at level $l+1$, if the weight of $x$ is contained in the $j$th child range of $[a,b],$
then $x_l \in T_l$ is labeled $j.$
Each $T_l$ is represented by \Cref{lemma:smallAlphabetTrees} in $n(\lg f + 2) + \smallO(n)$ bits,
so the total space cost of all the $T_l$'s is $n\lg n + (2n+\smallO(n))\log_f n$ bits.
When $f = \omega(1)$, this space cost is $n + \smallO(n)$ words.
This completes the outline of hierarchical tree extraction.
Henceforth, we shorthand as $\mathscr{T}_v$ the
extraction from $T$ of the nodes with weights in $v$'s range, for a node $v$ of the range tree. 
The following lemma maps $x_l$ to $x_{l+1}$:

\begin{lemma}[\cite{He:2016:DSP:2983296.2905368}]\label{lemma:rangetreemapping}
  Given a node $x_l \in T_l$ and the range $[a,b]$ of level $l$ containing the weight of $x$, 
  node $x_{l+1} \in T_{l+1}$ can be located in $\O(1)$ time, for any $l \in [h-2]$.
\end{lemma}

Later, Chan et al.~\cite{DBLP:journals/algorithmica/ChanHMZ17} augmented this representation with 
{\textit{ball-inheritance data structure}} to map an arbitrary $x_l$ back to $x$: 

\begin{lemma}[\cite{DBLP:journals/algorithmica/ChanHMZ17}]\label{lemma:conceptualTreeBallInheritance}
    Given a node $x_l \in T_l,$ where $1 \leq l < h,$ the node
	$x \in T$ that corresponds to $x_l$ can be found using $\O(n\lg{n}\cdot{}\mathtt{s}(n))$ bits of additional
	space and $\O(\mathtt{t}(n))$ time, 
	where
        (a) $\mathtt{s}(n) = \O(1)\text{ and }\mathtt{t}(n) = \O(\lg^{\epsilon}n)$; or
        (b) $\mathtt{s}(n) = \O(\lg\lg{n})\text{ and }\mathtt{t}(n) = \O(\lg\lg{n})$; or
        (c) $\mathtt{s}(n) = \O(\lg^{\epsilon}n)\text{ and }\mathtt{t}(n) = \O(1).$
\end{lemma}

\subparagraph*{Path Minimum in (Scalarly-)Weighted Trees.}
In a weighted tree, path minimum query asks for the node with the smallest weight in the given path.
We summarize the best result on path minimum; 
in it, ${\alpha}(m,n)$ and ${\alpha}(n)$ are the inverse-Ackermann functions:

\begin{lemma}[\cite{DBLP:journals/algorithmica/ChanHMZ17}]\label{lemma:lemmaPathMinimum}
	An ordinal tree $T$ on $n$ weighted nodes can be indexed
    (a) using $\O(m)$ bits of space to support path minimum queries
    in $\O({\alpha}(m,n))$ time and $\O({\alpha}(m,n))$ accesses
    to the weights of nodes, for any integer $m \geq n;$
    or (b) using $2n+\smallO(n)$ bits of space to support
    path minimum queries in $\O({\alpha}(n))$ time and $\O({\alpha}(n))$
    accesses to the weights of nodes.
    In particular, when $m = \Theta(n\lg^{**}{n}),$\footnote{
    $\lg^{**}{n}$ stands for the number of times
    an iterated logarithm function $\lg^{*}$ needs
    to be applied to $n$ in order for the result to become at most $1$.}
    one has ${\alpha}(m,n) = \O(1),$ and therefore
    (a) includes the result that
    $T$ can be indexed in $\O(n\lg^{**}{n})$ bits
    of space to support path minimum queries in $\O(1)$ time
    and $\O(1)$ accesses to the weights of nodes.
\end{lemma}

\section {Reducing to Lower Dimensions}\label{section:masterLemmaBinaryCase}
This section presents a general framework for reducing the problem of answering a 
$d$-dimensional query to the same query problem in $(d-1)$ dimensions,
by generalizing the standard technique of range tree
decomposition for the case of tree topologies weighted with multidimensional vectors.
To describe this framework, we introduce a {\textit{$d$-dimensional semigroup path sum query problem}} 
which is a generalization of all the query problems we consider in this paper. 
Let $(G,\oplus)$ be a semigroup and $T$ a tree on $n$ nodes, 
in which each node $x$ is assigned a $d$-dimensional weight vector $\pnt{w}(x)$ 
and a semigroup element $g(x),$ with the semigroup sum operator denoted as $\oplus.$
Then, in a $d$-dimensional semigroup path sum query, we are given a path $P_{x,y}$ in $T,$ 
an orthogonal query range $Q$ in $d$-dimensional space, and we are asked 
to compute $\sum_{z \in P_{x,y}\text{ and }\pnt{w}(z) \in Q}g(z)$.
When the weight vectors of the nodes and the query range are both from a $\dimcl{d'}{d}{\eps}$-dimensional space,
the {\textit{$\dimcl{d'}{d}{\eps}$-dimensional semigroup path sum query problem}} is defined analogously.

The following lemma presents our framework for solving a {\textit{$d$-dimensional semigroup path sum query problem}};
its counterpart in $\dimcl{d'}{d}{\eps}$-dimensional space is given in \Cref{section:masterLemmaGeneral}.

\begin{lemma}\label{lem:reductionBinaryCase}
    Let $d$ be a positive integer constant.
    Let $G^{(d-1)}$ be an $\mathtt{s}(n)$-word data structure for a
    $(d-1)$-dimensional semigroup path sum problem of size $n$. 
    Then, there is an $\O(\mathtt{s}(n)\lg{n}+n)$-word data structure $G^{(d)}$ for a $d$-dimensional semigroup path sum problem of size $n,$ 
    whose components include $\O(\lg n)$ structures of type $G^{(d-1)}$, each of which is constructed over a tree on $n+1$ nodes.
    Furthermore, $G^{(d)}$ can answer a $d$-dimensional semigroup path sum query by performing 
    $\O(\lg n)$ $(d-1)$-dimensional queries using these components and returning the semigroup sum of the answers. 
    Determining which queries to perform on structures of type $G^{(d-1)}$ requires $\O(1)$ time per query.
    \footnote{
        It may be tempting to simplify the statement of the lemma by defining $\mathtt{t}(n)$ as the query time of $G^{(d-1)}$ 
        and claiming that $G^{(d)}$ can answer a query in $\O(\mathtt{t}(n)\lg n)$ time. 
        However, this bound is too loose when applying this lemma to reporting queries.
    }
\end{lemma}
\begin{proof}
  We define a conceptual range tree $R$ with branching factor $2$ 
  over the $d$th weights of the nodes of $T$ and represent it using hierarchical tree extraction as in \Cref{section:preliminaries}.
  For each level $l$ of the range tree, we define a tree $T_l^*$ with the same topology as $T_l$.
  We assign $(d-1)$-dimensional weight vectors and semigroup elements to each node, $x'$, in $T_l^*$ as follows.
  If $x'$ is not the dummy root, 
  then $\pnt{w}(x')$ is set to be $(w_1(x),\ldots,w_{d-1}(x))$,
  where $x$ is the node of $T$ corresponding to $x'.$
  We also set $g(x') = g(x)$. 
  If $x'$ is the dummy root, then its first $(d-1)$ weights are $-\infty,$
  while $g(x')$ is set to an arbitrary element of the semigroup. 
  We then construct a data structure, $G_l$, of type $G^{(d-1)}$, over $T_l^*$.
  The data structure $G^{(d)}$ thus comprises the structures $T_l$ and $G_l$, over all $l.$
  The range tree has $\O(\lg n)$ levels, each $T_l^*$ has $n+1$ nodes, 
  and the $G_l$s are the $\O(\lg n)$ structures of type $G^{(d-1)}$ referred to in the statement.
  As all the structures $T_l$ occupy $n+\smallO(n)$ words, $G^{(d)}$ occupies $\O(\mathtt{s}(n)\lg{n}+n)$ words.

  Next we show how to use $G^{(d)}$ to answer queries.
  Let $P_{x,y}$ be the query path and $Q = \prod_{j=1}^{d}[q_j,q_j']$ be the query range.
  To answer the query, we first decompose $P_{x,y}$ into $A_{x,z}$, $\{z\},$ and $A_{y,z},$ where $z$ is the lowest common ancestor of $x$ and $y$,
  found in $\O(1)$ time via $\mathtt{LCA}$ in $T_1$. 
  It suffices to answer three path semigroup sum queries using each subpath and $Q$ as query parameters, 
  as the semigroup sum of the answers to these queries is the answer to the original query. 
  Since the query on subpath $\{z\}$ reduces to checking whether $\pnt{w}(z) \in Q,$ 
  we show how to answer the query on $A_{x,z}$; the query on $A_{y,z}$ is then handled similarly. 
  To answer the query on $A_{x,z}$, we perform a standard top-down traversal in the range tree
  to identify up to two nodes 
  at each level representing ranges that contain exactly one of $q_{d}$ or $q_{d}'.$
  Let, thus, $v$ be the node that we are visiting, in the range tree $R.$
  We maintain {\textit{current nodes}}, $x_v$ and $z_v$ (initialized
  as respectively $x$ and $z$) local to the current level $l;$
  they are the nodes in $T_l$ that correspond to $\mathscr{T}_{v}$-{\astralbody} of 
  the original query nodes $x$ and $z.$
  Nodes $x_v$ and $z_v$ are kept up-to-date in $\O(1)$ time as we descend the levels of the range tree.
  Namely, when descending to the $j$th ($j \in \{0,1\}$) child of the node $v,$
  we identify, via \Cref{lemma:rangetreemapping},
  the corresponding nodes in $T_{l+1},$ for nodes $\mathtt{level\_anc}_j(T_{l},x_v,1)$ and $\mathtt{level\_anc}_j(T_{l},z_v,1).$
  
  For each node $v$ identified at each level $l$, such that $v$'s range contains
  $q_d$ but not $q_d',$ we check if it is its left child-range that contains $q_d.$
  If so, we perform a $(d-1)$-dimensional semigroup range sum query 
  with the following parameters:
  (i) the query range $[q_1, q_1']\times[q_2, q_2']\times\ldots\times[q_{d-1}, q_{d-1}']$
  (i.e. we drop the last range); and
  (ii) the query path is $A_{x_{u},z_{u}}$, where $x_{u}$ and $z_{u}$ are the nodes in $T_{l+1}$ corresponding to the 
  $\mathscr{T}_{u}$-{\astralbody}s of $x$ and $z,$
  with $u$ being the right child of $v;$ 
  this is analogous to updating $x_v$ and $z_v,$
  i.e. applying \Cref{lemma:rangetreemapping} to nodes $\mathtt{level\_anc}_1(T_l,x_v,1)$, $\mathtt{level\_anc}_1(T_l,z_v,1).$
  For each node whose range contains $q_d'$ but not $q_d,$
  a symmetrical procedure is performed by considering its left child.

  The semigroup sum of the answers to these $\O(\lg{n})$ queries is the answer to the original query.
\end{proof}

\section{Space Reduction Lemma for Non-Constant Branching Factor}\label{section:masterLemmaGeneral}
This section presents a general framework for reducing the problem of answering a $\dimcl{d'}{d}{\eps}$-dimensional query 
to the same query problem in $\dimcl{d'-1}{d}{\eps}$ dimensions,
by generalizing the approach of~\cite{DBLP:conf/isaac/JaJaMS04} for the case of trees weighted with multidimensional vectors.
\begin{lemma}\label{lem:reductionGeneral}
    Let $d$ and $d'$ be positive integer constants such that $d' \leq d$, and $\eps$ be a constant in $(0,1).$
    Let $G^{(d'-1)}$ be an $s(n)$-word data structure for a $\dimcl{d'-1}{d}{\eps}$-dimensional semigroup path sum problem of size $n$. 
    Then, there is an $\O(s(n)\lg n/\lg\lg n+n)$-word data structure $G^{(d')}$ for a $\dimcl{d'}{d}{\eps}$-dimensional semigroup path sum problem 
	of size $n,$ whose components include $\O(\lg n / \lg\lg n)$ structures of type $G^{(d'-1)}$, 
	each of which is constructed over a tree on $n+1$ nodes.
    Furthermore, $G^{(d')}$ can answer a $\dimcl{d'}{d}{\eps}$-dimensional semigroup path sum query by performing 
	$\O(\lg n / \lg\lg n)$ $\dimcl{d'-1}{d}{\eps}$-dimensional queries using these components and 
	returning the semigroup sum of the answers. 
  	Determining which queries to perform on structures of type $G^{(d'-1)}$ requires $\O(1)$ time per query. 
\end{lemma}
\begin{proof}
  We define a conceptual range tree over the $d'$th weights of the nodes of $T$ and represent it using hierarchical tree extraction 
  as in \Cref{section:preliminaries}.
  For each level $l$ of the range tree, we define a tree $T_l^*$ with the same topology as $T_l$.
  We assign $\dimcl{d'-1}{d}{\eps}$-dimensional weight vectors and semigroup elements to each node, $x'$, in $T_l^*,$ as follows.
  If $x'$ is not the dummy root, 
  then $\pnt{w}(x')$ is set to be $(w_1(x),\ldots,w_{d'-1}(x),\lambda(T_l, x'),w_{d'+1}(x),\ldots,w_d(x))$,
  where $x$ is the corresponding node of $x'$ in $T,$
  and $\lambda(T_l, x')$ is the label assigned to $x'$ in $T_l$.
  We also set $g(x') = g(x)$. 
  If $x'$ is the dummy root, then its first $d'-1$ weights are $-\infty$ and last $d-d'+1$ weights 
  are $-\ceil{\lg^{\eps}}$, while $g(x')$ is set to an arbitrary element of the semigroup. 
  We further construct a data structure, $G_l$, of type $G^{(d'-1)}$, over $T_l^*$.
  The data structure $G^{(d')}$ then comprises the structures $T_l$ and $G_l$, over all $l.$ 
  The range tree has $\O(\lg n /\lg \lg n)$ levels and each $T_l^*$ has $n+1$ nodes, 
  and the structures $G_l$ are the $\O(\lg n /\lg \lg n)$ structures of type $G^{(d'-1)}$ referred to in the statement.
  As all the $T_l$s occupy $n+\smallO(n)$ words, $G^{(d')}$ occupies $\O(s(n)\lg n/\lg\lg n+n)$ words.

  Next we show how to use $G^{(d')}$ to answer queries.
  Let $P_{x,y}$ be the query path and $Q = \prod_{j=1}^{d}[q_j, q_j']$ be the query range.
  As discussed in the proof of \Cref{lem:reductionBinaryCase}, it suffices to describe
  the handling of the path $A_{x,z}$, where $z$ is the lowest common ancestor of $x$ and $y$.

  To answer the query on $A_{x,z}$, we perform a top-down traversal in the range tree to identify the up to two nodes 
  at each level representing ranges that contain at least one of $q_{d'}$ and $q_{d'}'$.
  For each node $v$ identified at each level $l$, we perform a $\dimcl{d'-1}{d}{\eps}$-dimensional 
  semigroup range sum query with parameters computed as follows:
    (i) the query path is $P_{x_v,z_v}$, where $x_v$ and $z_v$ are the nodes in $T_l$ corresponding to the $\mathscr{T}_v$-{\astralbody}s of $x$ and $z$; and
    (ii) the query range is $Q_v = [q_1, q_1']\times[q_2, q_2']\times\ldots\times[q_{d'-1}, q_{d'-1}']\times[i_v..j_v]\times[q_{d'+1}, q_{d'+1}']\times\ldots\times[q_d, q_d']$,
  such that the children of $v$ representing ranges that are entirely within $[q_{d'}, q_{d'}']$ 
    are children $i_v, i_v+1, \ldots, j_v$ (child $i$ refers to the $i$th child); no queries are performed if such children do not exist.
  The semigroup sum of these $\O(\lg n /\lg\lg n)$ queries is the answer to the original query.
  It remains to show that the parameters of each query are computed in $\O(1)$ time per query.
  By \Cref{section:preliminaries}, $i_v$ and $j_v$ are computed in $\O(1)$ time via simple arithmetic, which is sufficient to determine $Q_v$.
  Nodes $x_v$ and $z_v$ are computed in $\O(1)$ time each time we descend down a level in the range tree:
    Initially, when $v$ is the root of the range tree, $x_v$ and $z_v$ are nodes $x$ and $z$ in $T_1$.
    When we visit a child, $v_j$, of $v$ whose range contains at least one of $q_{d'}$ and $q_{d'}'$, 
    we compute (via \Cref{lemma:rangetreemapping}) $x_{v_j}$ as the node in $T_{l+1}$ 
    corresponding to the node $\mathtt{level\_anc}_j(T_l, x_v,1)$ in $T_l$, which uses constant time.
    Node $z_{v_j}$ is located similarly.
\end{proof}

\section{Ancestor Dominance Reporting}\label{section:ancestralReportingSection}
In \Cref{lemma:lemma3SidedReporting}
we solve the $\dimcl{1}{d}{\eps}$-dimensional {\textit{path dominance reporting problem}},
which asks one to enumerate the nodes in the query path whose weight vectors dominate the query vector. 
The strategy employed in \Cref{lemma:lemma3SidedReporting}
is that of zooming into the extraction dominating the query point in the last $(d-1)$ weights,
and therein reporting the relevant nodes based on the $1$st weight and tree topology only.
\begin{lemma}\label{lemma:lemma3SidedReporting}
    Let $d \geq 1$ be a constant integer and $0 < \eps < \frac{1}{d-1}$ be a constant number.
	A tree $T$ on $m \leq n$ nodes, 
	in which each node is assigned a $\dimcl{1}{d}{\eps}$-dimensional weight vector, 
    can be represented in $m+\smallO(m)$ words, so that a path dominance 
	reporting query can be answered in $\O(1+k)$ time, 
	where $k$ is the number of the nodes reported.
\end{lemma}
\begin{proof}
  We represent $T$ using \Cref{lemma:smallAlphabetTrees}.
  For any $\dimcl{0}{d-1}{\eps}$-dimensional vector $\pnt{g}= (g_1,g_2,\ldots,g_{d-1})$,  
  we consider a conceptual scalarly-weighted tree $E_{\pnt{g}}$ by first extracting the node set 
  $G = \{x\,|\,x \in T \text{ and }\pnt{w}_{2,d}(x)\succeq\pnt{g}\}$ from $T$. 
  The weight of a non-dummy node in $E_{\pnt{g}}$ is the $1$st weight of its $T$-source. 
  If $E_{\pnt{g}}$ has a dummy root, then its weight is $-\infty$. 
  
  Instead of storing $E_{\pnt{g}}$ explicitly, we create the following structures, 
  the first two of which are built for any possible $\dimcl{0}{d-1}{\eps}$-dimensional vector $\pnt{g}$:

  \begin{itemize}
		  \item A $0/1$-labeled tree $T_{\pnt{g}}$ (using \Cref{lemma:smallAlphabetTrees}) with the topology of $T$, 
			  in which a node $u$ has label $1$ {\textit{iff}} $u$ is extracted when constructing $E_{\pnt{g}}$;
		  \item A succinct index $I_{\pnt{g}}$ for path maximum queries in $E_{\pnt{g}}$
			   (using \Cref{lemma:lemmaPathMinimum}(a)); 
		  \item An array $W_1$ where $W_1[x]$ stores the $1$st weight of the node $x$ in $T$; 
		  \item A table $C$ which stores pointers to $T_{\pnt{g}}$ and $I_{\pnt{g}}$ for each possible $\pnt{g}$.
  \end{itemize}

For any node $x'$ in $E_{\pnt{g}}$, its $T$-source $x$ can be computed using 
$x = \mathtt{pre\_select}_1(T_{\pnt{g}},x')$. 
Then, the weight of $x'$ is $W_1[x]$.
With this $\O(1)$-time access to node weights in $E_{\pnt{g}}$, by 
\Cref{lemma:lemmaPathMinimum}
we can use $I_{\pnt{g}}$ to answer path maximum queries in $E_{\pnt{g}}$ in $\O(1)$ time.

We now show how to answer a path dominance reporting query in $T$. 
Let $P_{x,y}$ and $\pnt{q} = (q_1,q_2,\ldots,q_d)$ 
be respectively the path and weight vector given as query parameters.
First, we use $C$ to locate $T_{\pnt{q'}}$ and $I_{\pnt{q'}}$, where $\pnt{q'} = \pnt{q}_{2,d}$.
As discussed in the proof of \Cref{lem:reductionBinaryCase}, it suffices to show how to answer the query
with $A_{x,z}$ as the query path, where $z = \mathtt{LCA}(T,x,y).$
To that end, we fetch the $T_{\pnt{q'}}$-{\astralbody}, $x'$, of $x,$ as 
$x' = \mathtt{pre\_rank}_{1}(T_{\pnt{q'}}, \mathtt{level\_anc}_{1}(T_{\pnt{q'}},x,1))$, and
analogously the {\astralbody}, $z'$, of $z.$ 
Next, $I_{\pnt{q'}}$ locates a node $t' \in A_{x',z'}$ with the maximum weight.
If the weight of $t'$ is less than $q_1$, then no node in $A_{x,y}$ can possibly have a weight vector dominating $\pnt{q}$, 
and our algorithm is terminated without reporting any nodes. 
Otherwise, the $T$-source $t$ of $t'$ is located as $t = \mathtt{pre\_select_1}(T_{\pnt{q'}},t').$
The node $t \in T$ then claims the following two properties:
(i) as $T_{\pnt{q'}}$ contains a node corresponding to $t,$ one has $\pnt{w}_{2,d}(t)\succeq\pnt{q'};$ and
(ii) as $w_1(t)$ equals the weight of $t',$ it is at least $q_1.$
We therefore have that $\pnt{w}(t) \succeq \pnt{q}$ and hence report $t$. 
Afterwards, we perform the same procedure recursively on paths $A_{x',t'}$ and $A_{s',z'}$ in $E_{\pnt{q'}}$, 
where $s' = \mathtt{pre\_rank}_{1}(T_{\pnt{q'}},\mathtt{level\_anc_1}(T_{\pnt{q'}},t,1)).$

To analyze the running time, the key observation is that we perform path maximum queries using $I_{\pnt{q'}}$ at most $2k+1$ times.
Since both each query itself and the operations performed to identify the query path use $\O(1)$ time, our algorithm runs in $\O(1+k)$ time. 

To analyze the space cost, we observe that $W_1$ occupies $m$ words. 
The total number of possible $\dimcl{0}{d-1}{\eps}$-dimensional vectors is $\O(\lg^{(d-1)\eps}{n})$.
Since each $T_{\pnt{g}}$ uses $\O(m)$ bits and each $I_{\pnt{g}}$ uses $\O(m\lg^{**}{m})$ bits, 
the total space space cost of storing $T_{\pnt{g}}$'s and  $I_{\pnt{g}}$'s for all possible $\pnt{g}$'s is 
$\O((m+m\lg^{**}{m})\lg^{(d-1)\eps}{n}) 
	= \O(m\lg^{**}{m}\lg^{(d-1)\eps}{n}) \le \O(m\lg^{**}{n}\lg^{(d-1)\eps}{n})= \smallO(m\lg n)$ 
bits for any constant $0 < \eps < 1/(d-1),$ which is $\smallO(m)$ words.
Furthermore, $C$ stores $\O(\lg^{(d-1)\eps}{n})$ pointers. To save the space cost of each pointer, 
we concatenate the encodings of 
all the $T_{\pnt{g}}$s and $I_{\pnt{g}}$s and store them in a memory block of $\smallO(m\lg{n})$ bits.
Thus, each pointer stored in $C$ can be encoded in $\O(\lg(m\lg{n}))$ bits,
and the table $C$ thus uses 
$\O((\lg{m}+\lg\lg{n})\log^{(d-1)\eps}{n}) = \O(\lg{m}\log^{(d-1)\eps}{n})+\O(\lg\lg{n}\log^{(d-1)\eps}{n}) 
    = \smallO(\lg{m}\lg{n})+\smallO(\lg{n}) = \smallO(\lg{m}\lg{n})$ bits for any constant $0 < \eps < 1/(d-1)$, which is $\smallO(\lg{m})$ words.
Finally, the encoding of $T$ using \Cref{lemma:smallAlphabetTrees} is $2m+\smallO(m)$ bits.
Therefore, the total space cost is $m+\smallO(m)$ words. 
\end{proof}

We next design a solution to the $2$-dimensional ancestor dominance reporting problem,
by first generalizing the notion of $2$-dominance in Euclidean space 
to weighted trees. More precisely, in a tree $T$ in which each node is assigned a $d$-dimensional weight vector, 
we say that a node $x$ {\textit{$2$-dominates}} another node $y$ {\textit{iff}} $x \in \mathcal{A}(y)$ and $w_1(x) \geq w_1(y)$.
Then a node $x$ is defined to be {\textit{$2$-maximal}} {\textit{iff}} no other node in $T$ $2$-dominates $x$. 

The following property is then immediate:
Given a set, $X$, of $2$-maximal nodes, let $T_X$ be the corresponding extraction from $T$.
Let the weight of a node $x' \in T_X$ be the $1$st weight of its $T$-source $x$.
Then, in any upward path of $T_X$, the node weights are strictly decreasing.
In such a tree as $T_X$, the {\textit{weighted ancestor problem}}~\cite{fm1996} is defined.
In this problem, one is given a weighted tree with monotonically decreasing
node weights along any upward path.
We preprocess such a tree to answer {\textit{weighted ancestor queries}}, which,
for any given node $x$ and value $\kappa$, 
ask for the highest ancestor of $x$ whose weight is at least $\kappa$.
Farach and Muthukrishnan~\cite{fm1996} presented an $\O(n)$-word solution that answers 
this query in $\O(\lg\lg{n})$ time, for an $n$-node tree weighted over $[n]$. 
With an easy reduction we can further achieve the following result:

\begin{lemma}\label{lemma:hpdOnMab}
  Let $T$ be a tree on $m \leq n$ nodes, in which each node is assigned a weight from $[n]$. 
  If the node weights along any upward path are strictly decreasing, then $T$ can be represented using
  $\O(m)$ words to support weighted ancestor queries in $\O(\lg\lg{n})$ time. 
\end{lemma}
\begin{proof}
Let $W$ be the set of weights actually assigned to the nodes of $T$. 
We replace the weight, $h$, of any node $x$ in $T$ 
by the rank of $h$ in $W$, which is in $[m]$. 
We then represent the resulting tree $T'$ in $\O(m)$ words to support 
a weighted ancestor query in $T'$ in $\O(\lg\lg{m})$ time~\cite{fm1996}.
We also construct a $y$-fast trie~\cite{journals/ipl/Willard83}, $Y$, on the elements of $W$; 
the rank of each element is also stored with this element in $Y$. $Y$ uses $\O(m)$ space.
Given a weighted ancestor query over $T$, we first find the rank, $\kappa$, of the query weight in $W$ in 
$\O(\lg\lg{n})$ time by performing a predecessor query in $Y$, and $\kappa$ is further 
used to perform a query in $T'$ to compute the answer. 
\end{proof}

To design our data structures, we define a conceptual range tree with 
branching factor $f = \ceil{\lg^{\eps}{n}}$ over the $2$nd weights 
of the nodes in $T$ 
and represent it using hierarchical tree extraction as in \Cref{section:preliminaries}.
Let $v$ be a node in this range tree. In $\mathscr{T}_v$, we assign to each node the weight vector of the 
$T$-source and call the resulting weighted tree $T(v)$. 
We then define $M(v)$ as follows: If $v$ is the root of the range tree, then $M(v)$ is the set of all the $2$-maximal nodes in $T$. 
Otherwise, let $u$ be the parent of $v$. Then a node, $t$, of $T(v)$ is in $M(v)$ {\textit{iff}} 
$t$ is $2$-maximal in $T(v)$ and its corresponding node in $T(u)$ is not $2$-maximal in $T(u)$.
Thus, for any node $x$ in $T$, there exists a unique node $v$ in the range tree such that there is a node in $M(v)$ corresponding to $x$. 

We further conceptually extract two trees from $T_l:$ 
(i) $M_l$ is an extraction from $T_l$ of the node set 
$\{x\,|\,x \in T_l \text{ and there exists a node } u \text{ at level } l \text{ of the range tree, s.t. } x \text{ has a }$\\
$\text{corresponding node in } M(u)\}$; while 
(ii) $N_l$ is an extraction from $T_l$ of the node set 
$\{x\,|\,x \in T_l \text{ and } \exists \text{ a node } v \text{ at level } l+1 \text{ of the range tree, s.t. } x \text{ has a corresponding node} \in M(v)\}$. 
$T_M(v)$ is the tree formed by extracting $M(v)$ from $T(v)$, and $T_G(v)$ is the tree formed by extracting from $T(v)$ 
the node set 
$\{x\,|\, x \in T(v) \text{ and there exists a child, } t \text{ of } v $\\
$\text{s.t. there is a node in } M(t) \text{ corresponding to } x\}$. 
Then, for each level $l$, we also create the following data structures 
(when defining these structures, we assume that the root, $r_l$, of $T_l$ corresponds to a 
dummy node $\eta$ in $T$ with weight vector $(-\infty,-\infty)$; 
the node $\eta$ is omitted when determining the rank space, preorder ranks, and depths in $T$):
\begin{itemize}
    \item $D_l$, a $1$-dimensional path dominance reporting structure (using \Cref{lemma:lemma3SidedReporting})
        over the tree obtained by assigning weight vectors to the nodes of $M_l$ as follows: 
        each node $x'$ of $M_l$ is assigned a scalar weight $w_2(x)$, 
        where $x$ is the node of $T$ corresponding to $x';$ 
\item $E_l$, a $1$-dimensional path dominance reporting structure (using \Cref{lemma:lemma3SidedReporting})
    over the tree obtained by assigning weight vectors to the nodes of $M_l$ as follows: each node $x'$ of $M_l$ is assigned 
    a scalar weight $w_1(x)$, where $x$ is the node of $T$ corresponding to $x';$
\item $F_l$, a $\dimcl{1}{2}{\eps}$-dimensional path dominance reporting structure (using \Cref{lemma:lemma3SidedReporting}) 
      over the tree obtained by assigning weight vectors to the nodes of $N_l$ as follows: each node $x'$ of $N_l$ is 
      assigned $(w_1(x),\kappa)$, where $x$ is the node of $T$ corresponding to $x'$, 
      and $\kappa$ is the label assigned to the node in $T_l$ corresponding to $x'$; 
\item $A_l$, a weighted ancestor query structure over $M_l$ (using \Cref{lemma:hpdOnMab}), when its nodes are assigned the $1$st weights of the corresponding nodes in $T$; 
\item $T_l'$, a $0/1$-labeled tree (using \Cref{lemma:smallAlphabetTrees}) with the topology of $T_l$, and a node is assigned $1$ {\textit{iff}} it is extracted when constructing $M_l;$
\item $T_l''$, a $0/1$-labeled tree (using \Cref{lemma:smallAlphabetTrees}) with the topology of $T_l$, and a node is assigned $1$ {\textit{iff}} it is extracted when constructing $N_l;$
\item $P_l$, an array where $P_l[x]$ stores the preorder number of the node in $T$ corresponding to a node $x$ in $M_l$. 
\end{itemize}
We now describe the algorithm for answering queries, and analyze its running time
and space cost:
\begin{lemma}\label{lemma:ancesto2dBinaryCase}
  A tree $T$ on $n$ nodes, in which each node is assigned a $2$-dimensional weight vector, 
  can be represented in $\O(n)$ words, so that an ancestor dominance 
  reporting query can be answered in $\O(\lg{n}+k)$ time, 
  where $k$ is the number of the nodes reported.
\end{lemma}
\begin{proof}
  Let $x$ and $\pnt{q} = (q_1,q_2)$ be the node and weight vector given as query parameters, respectively.
  We define $\Pi$ as the path in the range tree between and including the root and the leaf storing $q_2.$
  Let $\pi_l$ denote the node at level $l$ in this path. Then the root of the range tree is $\pi_1$.
  To answer the query, we perform a traversal of a subset of the nodes of the range tree, starting from $\pi_1.$
  The invariant maintained during this traversal is that a node $u$ of the range tree is visited {\textit{iff}} one of the following two conditions holds: 
  (i) $u = \pi_l$ for some $l;$ or (ii) $M(u)$ contains at least one node whose corresponding node in $T$ must be reported. 
  We now describe how the algorithm works when visiting a node, $v$, at level $l$ of this range tree, 
  during which we will show how the invariant is maintained. 
  Let $x_v$ denote the node in $T_l$ that corresponds to the $\mathscr{T}_{v}$-{\astralbody} of $x$; $x_v$ can be located in 
  constant time each time we descend down one level in the range tree, as described in the proof of \Cref{lem:reductionBinaryCase}.
  Our first step is to report all the nodes in the answer to the query that have corresponding nodes in $M(v)$. 
  There are two cases depending on whether $v = \pi_l$; this condition can be checked in constant time by determining whether
  $q_2$ belongs to the range represented by $v$. In either of these cases, we first
  locate the $M_l$-{\astralbody}, $x_v'$, of $x_v$ by computing 
  $x_v' = \mathtt{pre\_rank}_1(T_l', \mathtt{level\_anc}_1(T_l',x_v,1))$.

  If (i) holds, then the non-dummy ancestors of $x_v'$ in $M_l$ correspond to all the ancestors of $x$ in $T$ that have corresponding nodes in $M(v)$. 
  We then perform a weighted ancestor query using $A_l$ to locate the highest ancestor, $y$, of $x_v'$ in $M_l$ whose $1$st weight is at least $q_1$. 
  Since the $1$st weights of the nodes along any upward path in $M_l$ are decreasing, the $1$st weights of the nodes in path $P_{x_v', y}$ are  
  greater than or equal to $q_1$, while those of the proper ancestors of $y$ are strictly less. 
  Hence, by performing a $1$-dimensional path dominance reporting query in $D_l$ using $P_{x_v',y}$ as the query path and 
  $\pnt{q'} = (q_2)$ as the query weight vector, we can find all the ancestors of $x_v'$ whose corresponding nodes in $T$ have weight vectors dominating $\pnt{q}$. 
  Then, for each of these nodes, we retrieve from $P_l$ its corresponding node in $T$ which is further reported. 
  
  IF $v \neq \pi_l$, the maintained invariant guarantees that the $2$nd weights of the nodes in $M(v)$ are greater than $q_2$. 
  Therefore, by performing a $1$-dimensional path dominance reporting query in $E_l(s)$ using the path between (inclusive) $x_v'$ and the root of $M_l$ as the query path and 
  $\pnt{q''} = (q_1)$ as the query weight vector, we can find all the ancestors of $x_v'$ in $M_l$ whose corresponding nodes in $T$ have weight vectors dominating $\pnt{q}$. 
  By mapping these nodes to nodes in $T$ via $P_l$, we have reported all the nodes in the answer to the query that have corresponding nodes in $M(v)$. 

  After we handle both cases, the next task is to decide which children of $v$ we should visit. 
  Let $v_i$ denote the $i$th child of $v$. 
  We always visit $\pi_{l+1}$ if it happens to be a child of $v$.
  To maintain the invariant, for any other child $v_i$, we visit it {\textit{iff}} there exists at least one node in $M(v_i)$ 
  whose corresponding node in $T$ should be reported. 
  To find the children that we will visit, we locate the $N_l$-{\astralbody}, $x_v''$, of $x_v$ by computing 
  $x_v'' = \mathtt{pre\_rank}_1(T_l'', \mathtt{level\_anc}_1(T_l'',x_v,1))$.
  Then the non-dummy ancestors of $x_v''$ correspond to all the ancestors of $x$ in $T$ that have corresponding nodes in $\cup_{i=1, 2, \ldots}M(v_i)$. 
  We then perform a $\dimcl{1}{2}{\eps}$-dimensional path dominance reporting query in $F_l$ using the path between (inclusive) $x_v''$ and the root of $N_l$ 
  as the query path and $(q_1,\kappa+1)$ as the query weight vector if $\pi_{l+1}$ is the $\kappa$th child of $v$, and we set $\kappa = 0$ if $\pi_{i+1}$ is not a child of $v$. 
  For each node, $t$, returned when answering this query, if its $2$nd weight in $F_l$ is $j$, then $t$ corresponds to a node in $M(v_j)$. 
  Since the node corresponding to $t$ in $T$ should be included in the answer to the original query, we iteratively visit $v_j$ if we have not visited it before 
  (checked e.g. using an $f$-bit word to flag the children of $v$). 
  
  The total query time is dominated by the time used to perform queries using $A_l$, $D_l$, $E_l$ and $F_l$.
  We only perform one weighted ancestor query when visiting each $\pi_l$,  and this query is not performed when visiting other nodes of the range tree.
  Given the $\O(\lg n/\lg\lg n)$ levels of the range tree, all the weighted ancestor queries collectively use $\O(\lg\lg n \times (\lg n/\lg\lg n)) = \O(\lg n)$ time.
  Similarly, we perform one query using $D_l$ at each level of the range tree, and the query times summed over all levels is $\O(\lg n/\lg\lg n + k)$.
  Our algorithm guarantees that, each time we perform a query using $E_l$, we report a not-reported hitherto, non-empty subset of the nodes in the answer to the original query. 
  Therefore, the queries performed over all $E_l$'s use $\O(k)$ time in total. Querying the $F_l$-structures incurs $\O(k)$ time cost when visiting nodes not in $\Pi$, 
  and $\O(\lg n/\lg\lg n + k)$ time when visiting nodes in $\Pi$.
  Thus, the query times spent on all these structures throughout the execution of the algorithm sum up to $\O(\lg{n}+k)$.   

  We next analyze space cost of our data structures.
  As mentioned in \Cref{section:preliminaries}, all the $T_l$s occupy $n + \smallO(n)$ words.
  By \Cref{lemma:smallAlphabetTrees}, each $T_l'$ or $T_l''$ uses $3n + \smallO(n)$ bits, so over all $\lg n /\lg\lg n$ levels, 
  they occupy $\O(n\lg n/\lg\lg n)$ bits, which is $\O(n/\lg\lg n)$ words.
  As discussed earlier, we know that, for any node $x$ in $T$, there exists one and only one node $v$ 
  in the range tree such that there is a node in $M(v)$ corresponding to $x$.
  Furthermore, $M(v)$s only contain nodes that have corresponding nodes in $T.$
  Therefore, the sum of the sizes of all $M(v)$s is exactly $n$.
  Hence all the $P_l$'s have $n$ entries in total and thus uses $n$ words.
  By \Cref{lemma:lemma3SidedReporting}, the size of each $D_l$ in words is linear in the number of nodes in $M_l$.
  The sum of the numbers of nodes in $M_l$s over all levels of the range tree is equal to the sum of the sizes of all $M(v)$s 
  plus the number of dummy roots, which is $n + \O(\lg n/\lg\lg n)$.
  Therefore, all the $D_l$s occupy $\O(n)$ words.
  By similar reasoning, all the $E_l$s and $A_l$s occupy $\O(n)$ words in total.
  Finally, it is also true that, for any node $x$ in $T$, there exists a unique node $v$ in the range tree 
  such that there is a node in $N(v)$ corresponding to $x$.
  Thus, we can upper-bound the total space cost of all the $F_l$s by $\O(n)$ words in a similar way.
  All our data structures, therefore, use $\O(n)$ words.
\end{proof}

Further, we describe the data structure for $\dimcl{2}{d}{\eps}$-dimensional
ancestor dominance reporting, and analyze its time- and space-bounds:
\begin{lemma}\label{lemma:ancesto2de}
  Let $d \geq 2$ be a constant integer and $0 < \eps < \frac{1}{d-2}$ be a constant number.
  A tree $T$ on $n$ nodes, in which each node is assigned a $\dimcl{2}{d}{\eps}$-dimensional weight vector, 
  can be represented in $\O(n\lg^{(d-2)\eps}{n})$ words, so that an ancestor dominance 
  reporting query can be answered in $\O(\lg{n}+k)$ time, 
  where $k$ is the number of the nodes reported.
\end{lemma}
\begin{proof}
    In our design, for any $\dimcl{0}{d-2}{\eps}$-dimensional vector $\pnt{g},$ 
    we consider a conceptual scalarly-weighted tree 
    $E_{\pnt{g}}$ as the tree extraction from $T$ of the node set $\{x\,|\,x \in T \text{ and } \pnt{w}_{3,d}(x)\succeq\pnt{g}\}.$
    The weight of a node $x'$ in $E_{\pnt{g}}$ is the $2$-dimensional weight vector $\pnt{w}_{1,2}(x),$ 
    where $x$ the $T$-source of $x'.$
    If $E_{\pnt{g}}$ has a dummy root, then its weight is $(-\infty,-\infty).$ 
    Rather than storing $E_{\pnt{g}}$ explicitly, we follow the strategy in the proof of \Cref{lemma:lemma3SidedReporting} 
    and store a $0/1$-labeled tree $T_{\pnt{g}}$ for each possible $\pnt{g}$.
    $T_{\pnt{g}}$ is obtained from $T$ by assigning $1$-labels to the nodes of $T$ extracted when constructing $E_{\pnt{g}}$.
    We also maintain arrays $W_1$ and $W_2$ storing respectively the $1$st and $2$nd  weights of all nodes of $T,$ in preorder,
    which enables accessing the weight of an arbitrary node of $E_{\pnt{g}}$ in $\O(1)$ time. 
    Let $n_{\pnt{g}}$ be the number of nodes in $E_{\pnt{g}}.$ We convert the node weights of each $E_{\pnt{g}}$ to rank space $[n_{\pnt{g}}].$
    For each such $E_{\pnt{g}},$ we build the $2$-dimensional ancestor dominance reporting data structure,
    $V_{\pnt{g}}$, from \Cref{lemma:ancesto2dBinaryCase}.
    Thus, the space usage of the resulting data structure is upper-bounded by $\O(n\lg^{(d-2)\eps}{n})$ words.

    Let $x$ and $\pnt{q} = (q_1,q_2,\ldots,q_d)$ be the node and weight vector given as query parameters, respectively.
    In $\O(1)$ time, we fetch the data structures pertaining to the range $\pnt{q'} = \pnt{q}_{3,d};$
    this way, all the weights $3$ through $d$ of the query vector have been taken care of, and all we need
    to consider is the tree topology and the first two weights, $q_1$ and $q_2,$ of the original query vector.
    We localize the query node $x$ to $E_{\pnt{q'}}$ via $x' = \mathtt{pre\_rank}_{1}(T_{\pnt{q'}},\mathtt{level\_anc}_1(T_{\pnt{q'}},x,1)),$
    and launch the query in $V_{\pnt{q'}}$ with $x'$ as a query node, having
    reduced the components of the query vector $(q_1,q_2)$ to the rank space of $E_{\pnt{q'}}$
    (the time- and space-bounds for the reductions are absorbed in the final bounds).
\end{proof}

Instantiating \Cref{section:masterLemmaBinaryCase} with $g(x) = \{x\}$ and the semigroup sum operator $\oplus$ as the set-theoretic union operator $\cup,$
\Cref{lem:reductionBinaryCase} iteratively applied to \Cref{lemma:ancesto2dBinaryCase} yields
\begin{theorem}\label{theorem:pathDominanceReportingBinaryCase}
  Let $d \geq 2$ be a constant integer.
  A tree $T$ on $n$ nodes, in which each node is assigned a $d$-dimensional weight vector,  
  can be represented in $\O(n\lg^{d-2}{n})$ words, so that an ancestor dominance 
  reporting query can be answered in $\O(\lg^{d-1}n+k)$ time, 
  where $k$ is the number of the nodes reported.
\end{theorem}
Analogously, \Cref{lem:reductionGeneral} (\Cref{section:masterLemmaGeneral})
that is the counterpart of \Cref{lem:reductionBinaryCase} when the range tree
has a non-constant branching factor $f = \O(\lg^{\eps}{n})$,
with \Cref{lemma:ancesto2de} 
which addresses $\dimcl{2}{d}{\eps}$-dimensional ancestor reporting,
together yield a different tradeoff:
\begin{theorem}\label{theorem:pathDominanceReportingGeneralCase}
  Let $d \geq 3$ be a constant integer.
  A tree $T$ on $n$ nodes, in which each node is assigned a $d$-dimensional weight vector,  
  can be represented in $\O(n\lg^{d-2+\eps}{n})$ words of space, so that an ancestor dominance 
  reporting query can be answered in $\O((\lg^{d-1} n)/(\lg\lg n)^{d-2}+k)$ time, 
  where $k$ is the number of the nodes reported. Here, $\eps \in (0,1)$ is 
  a constant.
\end{theorem}
\section{Path Successor}\label{section:pathRangeSuccessor}
We first solve the path successor problem when $d = 1,$
and extend the result to $d > 1$ via \Cref{lem:reductionBinaryCase}.

The topology of $T$ is stored using \Cref{lemma:smallAlphabetTrees}.
We define a binary range tree $R$ over $[n],$
and build the associated hierarchical tree extraction as in \Cref{section:preliminaries};
$T_{l}$ denotes the auxiliary tree built for each level $l$ of $R,$
and $\mathscr{T}_{v}$ denotes the tree extraction from $T$ associated with the range of node $v \in R.$
We represent $R$ using \Cref{lemma:smallAlphabetTrees}, and augment it with the
ball-inheritance data structure $\mathcal{B}$ from \Cref{lemma:conceptualTreeBallInheritance}(a),
as well as with
the data structure from the following
\begin{lemma}\label{lemma:markedView}
    Let $R$ be a binary range tree with topology encoded using \Cref{lemma:smallAlphabetTrees}, and augmented with
    ball-inheritance data structure $\mathcal{B}$ from \Cref{lemma:conceptualTreeBallInheritance}(a).
    With additional space of $\O(n)$ words,
    the node $x_{u,l}$ in $T_{l}$ corresponding to the $\mathscr{T}_{u}$-{\astralbody} of $x$
    can be found in $\O(\log^{\eps'}{n})$ time,
    where $\eps'$ is an arbitrary constant in $(0,1),$
    for an arbitrary node $x \in T$ and an arbitrary node $u \in R$ residing
    on a level $l.$
\end{lemma}
\begin{proof}
    For a node $u \in R$ at a level $l,$ and a node $x \in T,$
    the query can be thought of as a chain of 
    transformations $T \rightarrow \mathscr{T}_{u} \rightarrow T_{l}.$
    In the first transition, $T \rightarrow \mathscr{T}_u,$ 
    given an original node $x \in T,$
    we are looking for its $\mathscr{T}_u$-{\astralbody}, $x_u.$ That is, although $\mathscr{T}_u$
    is obtained from $T$ through a {\textit{series}} of extractions,
    the wish is to ``jump'' many successive extractions at once,
    as if $\mathscr{T}_u$ were extracted from $T$ {\textit{directly}}. 
    This would be trivial to achieve through storing a $0/1$-labeled tree per range $u,$
    if it were not for prohibitive space-cost -- number of bits
    more than quadratic in the number of nodes.
    One can avoid extra space cost altogether
    and use \Cref{lemma:rangetreemapping} directly to explicitly
    descend the hierarchy of extractions. In this case, the time cost
    is proportional to the height of the range tree, and hence becomes the bottleneck.
    
    In turn, in the $\mathscr{T}_u \rightarrow T_l$-transition,
    we are looking for the identity of $x_u$ in $T_l.$
	For this second transformation, we recall (from \Cref{section:preliminaries}) that $\mathscr{T}_u$ 
	is embedded within $T_l.$ Moreover, the nodes of $\mathscr{T}_u$ must lie contiguously in the preorder sequence of $T_l.$
	
    We overcome these difficulties with the following data structures.

    For $R,$ we maintain an annotation array $I,$
    such that $I[u]$ stores a quadruple $\langle{}a_u,b_u,s_u,t_u\rangle$ for an arbitrary node $u\in R$, 
    such that
    (i) the weight range associated with $u$ is $[a_u,b_u];$ and
    (ii) all the nodes of $T$ with weights in $[a_u,b_u]$ occupy precisely the preorder 
    ranks $s_u$ through $t_u$ in $T_{l}.$
    The space occupied by the annotation array $I$, which is $\O(\lg{n})$ bits summed over
    all the $\O(n)$ nodes of $R,$ is $\O(n)$ words.

    For
    each level $L \equiv 0\,\mathrm{mod}\, \ceil{\lg\lg{n}}$, which we call {\textit{marked}},
    we maintain a data structure enabling the direct 
    $T \rightarrow \mathscr{T}_u \rightarrow T_L$-conversion.
    Namely,
    for each individual node $u$ on marked level $L$ of $R,$ 
    we define a conceptual array $A_{u}$, which stores, 
    in increasing order, the (original) preorder ranks of all the nodes of 
    $T$ whose weights are in the range represented by $u$.
    Rather than maintaining $A_{u}$ explicitly, we store a succinct index, $S_{u}$, 
    for predecessor/successor search~\cite{grossi_et_al:LIPIcs:2009:1847} in $A_{u}$.
    Assuming the availability of a $\smallO(n^{\delta})$-bit universal table, where $\delta$ is a 
    constant in $(0, 1)$, given an arbitrary value in $[n]$, 
    this index can return the position of its predecessor/successor in 
    $A_{u}$ in $\O(\lg\lg{n})$ time plus accesses to $\O(1)$ entries of $A_{u}$. 
    The size of the index in bits is $\O(\lg\lg{n})$ times the number of entries in $A_{u}$.
    For a fixed marked level $L,$ therefore, all the $S_{u}$-structures sum up to $\O(n\lg\lg{n})$ bits.
    There being $\O(\lg{n}/\lg\lg{n})$ marked levels, the total space cost
    for the $S_{u}$-structures over all the entire tree $R$ is $\O(n)$ words.

    We now turn to answering the query using the data structures built.
    Resolving the query falls into two distinct cases, depending
    on whether the level $l$, at which the query node $u$ resides, is marked or not.

    When the level $l$ is marked, we use the structures $S_{u}$ stored therein, directly.
    We adopt the strategy in~\cite{DBLP:journals/algorithmica/HeMZ14} to find $x_{u,l}.$ 
    First, for an arbitrary index $i$ to $A_u,$ we observe that node $A_{u}[i] \in T$ corresponds to node $(s_u+i-1)$ in $T_{l}.$
    We thus fetch $\langle{}a_u,b_u,s_u,t_u\rangle{}$ from $I[u].$ Then the predecessor $p \in A_{u}$ of $x$
    is obtained through $S_{u}$ via an $\O(\lg\lg{n})$ query and $\O(1)$ calls to the $\mathcal{B}$-structure,
    which totals $\O(\lg^{\eps'}{n})$ time. We then determine the lowest common ancestor $\chi \in T$ 
    of $x$ and $p,$ in $\O(1)$ time. If the weight of $\chi$ is in $[a_u,b_u],$ then it 
    must be present in $A_{u}$ by the
    latter's very definition.
    By another predecessor query, therefore, we can find the position, $j$, 
    of $\chi$ in $A_{u}$, and $(s_u+j-1)$ is the sought $x_{u,l}.$
    Otherwise, a final query to $S_{u}$ returns the successor $\chi'$ in $A_{u}$ of $\chi.$
    Let $\kappa$ be the position of $\chi'$ in $A_{u}$. 
    Then the parent of the node $(s_u+\kappa-1)$ in $T_{l}$ is $x_{u,l}.$
    We perform a constant number of predecessor/successor queries, and correspondingly
    a constant number of calls to the ball-inheritance problem. The time complexity
    is thus $\O(\lg^{\eps'}{n}).$

    When the level $l$ is not marked, we ascend to the lowest ancestor $u'$ of 
    $u$ residing on a marked level $l',$ and reduce the problem to the previous case.
    More precisely, via the navigation operations ($\mathtt{level\_anc}()$
    to move to a parent, and $\mathtt{depth}()$ to determine the status of a level)
    available through $R$'s encoding, we climb up at most $\ceil{\lg\lg{n}}$
    levels to the closest marked level $l'.$ Let $u'$ be therefore
    the ancestor of $u$ found on that marked level $l'$.
    We find the answer to the original query, as if the query
    node were $u';$ that is, we find the node $x'_{u',l'}$ in $T_{l'}$ 
    that corresponds to the $\mathscr{T}_{u'}$-{\astralbody}
    of the original query node $x$ in $T.$
    Let us initialize a variable $\chi$ to be $x_{u',l'}.$
    We descend down to the original level $l,$ back to the original
    query node $u,$ all the while
    adjusting the node $\chi$ as we move down a level, analogously
    to the proof of \Cref{lem:reductionBinaryCase}.
    As we arrive, in time $\O(\lg\lg{n})$, at node $u,$
    the variable $\chi$ stores the answer, $x_{u,l}.$

    In both cases, the term $\O(\lg^{\eps'}{n})$ dominates the time complexity,
    as climbing to/from a marked level is an additive term of $\O(\lg\lg{n}).$
    Therefore, a query is answered in $\O(\lg^{\eps'}{n})$ time.
\end{proof}

Finally, each $T_{l}$ is augmented with succinct indices ${m}_{l}$ 
(resp. ${M}_{l}$) from \Cref{lemma:lemmaPathMinimum}(b),
for path minimum (resp. path maximum) queries.
As weights of the nodes of $T_{l},$ the weights of their corresponding nodes in $T$ are used.

We now describe the algorithm for answering queries
and analyze its running time, as well as give the space cost of the built data structures:

\begin{lemma}\label{lemma:rangeSuccessorBaseCase01}
    A scalarly-weighted tree $T$ on $n$ 
    nodes can be represented in $\O(n)$ words,
    so that a path successor query
    is answered in $\O(\lg^{\eps}{n})$ time, where $\eps \in (0,1)$ is a constant.
\end{lemma}
\begin{proof}
Let $x,y$ and $Q=[q_1,q'_1]$ be respectively the nodes and the orthogonal range given as query's parameters.
Appealing to the proof of \Cref{lem:reductionBinaryCase}, we focus only on 
the path $A_{x,z},$ where $z$ is $\mathtt{LCA}(T,x,y).$ 
We locate in $\O(1)$ time the leaf $L_{q_1}$ of $R$ that corresponds to the singleton range $[q_1,q_1].$
Let $\Pi$ be the root-to-leaf path to $L_{q_1}$ in $R;$ let $\pi_l$ be the node at level $l$ of $\Pi.$
We binary search in $\Pi$ for the deepest node $\pi_f \in \Pi$ whose associated 
extraction $\mathscr{T}_{\pi_f}$ contains the node corresponding to the answer to the given query.

We initialize two variables: $high$ as $1$ so that $\pi_{high}$ is the root of $R$, and $low$
as the height of $R$ so that $\pi_{low}$ is the leaf $L_{q_1}.$ 
We first check if $\mathscr{T}_{\pi_{low}}$ already contains the answer,
by fetching the node $x'$ in $T_{low}$ corresponding to the $\mathscr{T}_{\pi_{low}}$-{\astralbody} of $x,$
using \Cref{lemma:markedView}.
If $x'$ exists, we examine its corresponding node $x''$ in $T$ (fetched via $\mathcal{B}$)
to see whether $x''$ is on $A_{x,z},$ 
by performing $\mathtt{depth}$ and $\mathtt{level\_anc}$
operations in $R;$ 
if it is, $x''$ is the final answer.
If not, this establishes the invariant of the ensuing search:
$\mathscr{T}_{\pi_{high}}$ contains a node corresponding to the answer, whereas $\mathscr{T}_{\pi_{low}}$ does not.
   
At each iteration, therefore, we set (via $\mathtt{level\_anc}$ in $R$) $\pi_{mid}$ to be the node mid-way from $\pi_{low}$
to $\pi_{high}.$ 
We then fetch the nodes $x',z'$ in $T_{mid}$ corresponding to the $\mathscr{T}_{\pi_{mid}}$-{\astralbody}s of respectively $x$ and $z$,
using \Cref{lemma:markedView}.
The non-existence of $x'$ or the emptiness of $A_{x',z'}$
sets $low$ to $mid,$ and the next iteration of the search ensues.
If $z'$ does not exist, $z'$ is set to the root of $T_{mid}.$
A query to the $M_{mid}$-structure then locates a node in $A_{x',z'}$ for which the $1$st weight, $\mu$, of its 
corresponding node in $T$ is maximized.
Accounting for the mapping of a node in $T_{mid}$ to its corresponding node in $T$ via $\mathcal{B}$,
this query uses $\O((\lg^{\eps'} n)\alpha(n))$ time. 
The variables are then updated as $high \leftarrow mid$ if $\mu \geq q_1$, and $low \leftarrow mid,$ otherwise.

Once $\pi_f$ is located, it must hold for $\pi_f$ that
(i) it is its left child that is on $\Pi$~\cite{DBLP:conf/swat/NekrichN12}; and
(ii) its right child, $v$, contains the query result, even though $v$ represents a range of values all larger than $q_1$.
When locating $\pi_f$, we also found the nodes in $T_{f}$ corresponding to the $\mathscr{T}_{\pi_{f}}$-{\astralbody}s of $x$ and $z$;
they can be further used to find the nodes in $T_{f+1}$, $x^*$ and $z^*$, corresponding to the $\mathscr{T}_{v}$-{\astralbody}s of $x$ and $z$.
We then use $m_{f+1}$ to find the node in $A_{x^{*},z^{*}}$ with minimum $1$st weight, whose 
corresponding node in $T$ is the answer.

The total query time is dominated by that needed for binary search.
Each iteration of the search is in turn dominated by the path maximum query in $T_{mid}$, which is $\O((\lg^{\eps'} n)\alpha(n)).$
Given the $\O(\lg{n})$ levels of $R$, the binary search has $\O(\lg\lg{n})$ iterations.
Therefore, the total running time is $\O(\lg\lg{n}\cdot{}\lg^{\eps'}{n}\cdot{}\alpha(n)) = \O(\lg^{\eps} n)$ if we choose $\eps' < \eps.$ 

To analyze the space cost, we observe that
    the topology of $T$, represented using \Cref{lemma:smallAlphabetTrees}, uses only $2n+\smallO(n)$ bits.
    As mentioned in \Cref{section:preliminaries}, all the structures $T_{l}$ occupy $\O(n)$ words.
    The space cost of the structure from \Cref{lemma:markedView}
    built for $R$ is $\O(n)$ words.
    The $\mathcal{B}$-structure occupies another $\O(n)$ words.
    The ${m}_{l}$- and ${M}_{l}$-structure occupy $\O(n)$ bits each,
    or $\O(n)$ words in total over all levels of $R.$ 
    Thus, the final space cost is $\O(n)$ words.
\end{proof}

\Cref{lemma:rangeSuccessorBaseCase01,lem:reductionBinaryCase}
yield the following
\begin{theorem}\label{theorem:pathSuccessor}
  Let $d \geq 1$ be a constant integer.
  A tree $T$ on $n$ nodes, in which each node is assigned a $d$-dimensional weight vector
  can be represented in $\O(n\lg^{d-1}{n})$ 
  words, so that 
  a path successor query can be answered in $\O(\lg^{d-1+\eps}{n})$
  time, for an arbitrarily small positive constant $\eps$. 
\end{theorem}
\begin{proof}
   We instantiate \Cref{section:masterLemmaBinaryCase} with $g(x) = x$ 
   and the semigroup sum operator $\oplus$ as $x\oplus{}y = \mathtt{argmin}_{\zeta = x,y}\{w_1(\zeta)\}.$ 
   \Cref{lem:reductionBinaryCase} applied to \Cref{lemma:rangeSuccessorBaseCase01} yields the space bound of
   $\O(n\lg^{d-1}{n})$ words
   and query time complexity of $\O(\lg^{d-1+\eps}{n}).$
\end{proof}

\section{Path Counting}\label{appendix:pathCounting}
Note that we can not directly apply previous approaches, 
such as e.g.~\cite{He:2016:DSP:2983296.2905368} as a base data structure
for the $1D$ case. Those being not $\dimcl{1}{d}{\eps}$-dimensional data structures,
na{\"i}ve way of accommodating the last $(d-1)$ weights
would inevitably incur an extra $\lg^{2(d-1)\eps}{n}$-factor
in space. To do better than that, we need a few more techniques,
notably, {\textit{tree covering.}}
Tree covering first appeared in~\cite{Geary:2006:SOT:1198513.1198516,DBLP:journals/talg/HeMS12,DBLP:journals/algorithmica/FarzanM14}
as a method of succinct representation of ordinal trees.
The original tree is split into {\textit{mini}}-trees, given a certain parameter $L$:
\begin{lemma}[\cite{DBLP:journals/algorithmica/FarzanM14}]\label{lemma:lemmaOnTreeCover01}
	A tree with $n$ nodes can be decomposed into $\Theta(n/L)$ subtrees
	of size at most $2L.$ These are pairwise disjoint aside from the subtree roots.
	Furthermore, aside from the edges stemming from the subtree roots,
	there is at most one edge per subtree leaving a node of a subtree
	to its child in another subtree.
\end{lemma}
Each of the mini-trees
in turn can be recursively decomposed into {\textit{micro}}-trees, with another parameter $L'.$
The idea is to choose the parameter $L'$ such that {\textit{intra}}-micro-tree queries
are executed in constant-time by virtue of a precomputed table $\mathcal{T}$ of
size $\smallO(n),$
indexed by micro-trees. Consequently, for any given node $x \in T,$
the solutions of~\cite{Geary:2006:SOT:1198513.1198516,DBLP:journals/talg/HeMS12,DBLP:journals/algorithmica/FarzanM14}
provide a constant-time access to the mini-tree $\tau$ and micro-tree $\tau'$
containing the node $x,$ as well as the address of the micro-tree $\tau'$ in the table $\mathcal{T},$
using $\O(n)$ bits of space, 
for suitably chosen parameters $L$ and $L'.$

The goal of this section is to design a data structure to solve the 
path counting problem 
when the nodes are assigned $\dimcl{0}{d}{\eps}$-dimensional weight vectors.
It turns out that when the weight vector of a node can be packed into $\smallO(\lg{n})$ bits,
counting queries can be executed in constant time. The key machinery used is tree covering.

Let $T$ be a given ordinal tree on $n$ nodes, each node of which is assigned
a $\dimcl{0}{d}{\eps}$-dimensional weight vector. Let $\rho$ be the root of $T.$
In our solution to the $\dimcl{0}{d}{\eps}$-dimensional path counting problem for $T,$
we set $c = \ceil{\lg^{\eps}n},$ and using \Cref{lemma:lemmaOnTreeCover01}
perform the decomposition of $T$ into mini-trees with parameter 
$L = c^{2d}\lg{n}.$ Each of the mini-trees is further subject to decomposition
into micro-trees with parameter $L' = c^{2d}.$
Each mini- or micro-tree $b$ stores
an array $b.cnt,$ indexed by a tuple from $([c]\times{}[c])^d,$ 
with the following contents
(henceforth let $r_b$ be the root of the mini- or micro-tree $b$):
\begin{itemize}
    \item for a mini-tree $b,$ an $\O(\lg{n})$-bit number $b.cnt[Q]$
          stores the answer to the path counting query with parameters $r_b,\rho,Q;$ 
          i.e., $b.cnt[Q]$ is the number
          of the nodes with weight vectors falling within the range $Q$
          on the path $A_{r_b,\rho}$ in $T,$ where $\rho$ is the root of $T;$
    \item for a micro-tree $b'$ inside a mini-tree $b,$
          an $\O(\lg\lg{n})$-bit number $b'.cnt[Q]$ stores the answer
          to the path counting query with parameters $r_{b'},r_{b},Q;$
          i.e., $b'.cnt[Q]$ is the number
          of the nodes with weight vectors falling within the range $Q$
          on the path $A_{r_{b'},r_{b}}.$
\end{itemize}
We also precompute a look-up table $D$ that is indexed by a quadruple
from the following Cartesian product:
\begin{itemize}
    \item all the possible micro-tree topologies $\tau,$ {\textit{times}}
    \item all possible assignments $\lambda$ of weight vectors to the nodes of $\tau,$ {\textit{times}}
    \item all nodes in $\tau$, {\textit{times}}
	\item all possible query orthogonal ranges $Q.$
\end{itemize} 
The entry $D[\tau,\lambda,x,Q]$ stores the answer
to the path counting query with parameters $x,r_{\tau},Q$ over a micro-tree
with topology $\tau,$ nodes of which are assigned weight vectors
from configuration $\lambda.$ More precisely, $\lambda$ is 
a labeling of the micro-tree $\tau$ with $\dimcl{0}{d}{\eps}$-dimensional
weight vectors, and $D[\tau,\lambda,x,Q]$ stores the number
of nodes on the path $A_{x,r_{\tau}}$ in $\tau,$ such
that their weight vectors belong to the range $Q.$

We now show how to use these data structure to answer queries.

\begin{lemma}\label{lemma:pathCountingQueryBaseCase0Time}
  The data structures in this section can answer a $\dimcl{0}{d}{\eps}$-dimensional 
  path counting query in $\O(1)$ time, for any constant integer $d \ge 1.$
\end{lemma}
\begin{proof}
Let $P_{x,y}$ and $Q$ be, respectively, the path and the orthogonal range given as the parameters to the query.
Using the notation from the proof of \Cref{lem:reductionBinaryCase}, and for the reasons given therein, 
we describe only how to answer the query over $A_{x,z},$ where $z = \mathtt{LCA}(T,x,y).$
We assume the encoding of $T$ as in \Cref{lemma:smallAlphabetTrees}, so the $\mathtt{LCA}$
operator is available; the space overhead is only $\O(n)$ bits, 
i.e. negligible with respect to the space bound we are ultimately aiming at.

We further notice that answering path counting query over $A_{x,z}$
is equivalent to answering two path counting queries, one over $A_{x,\rho}$
and another over $A_{z,\rho}$, and taking their arithmetic difference.
It is thus sufficient to describe the procedure of
answering the query over $A_{x,\rho},$ and analyze its running time,
as the query over $A_{z,\rho}$ can be handled similarly.

The key observation is that overall we perform a constant number of constant-time operations.
Indeed, using data structures of~\cite{DBLP:journals/algorithmica/FarzanM14},
we first identify, in $\O(1)$ time, the mini-tree $b$ and micro-tree $b'$ containing
the node $x,$ as well as the encoding of $b'.$
From the (disjoint) decomposition of the path
$A_{x,\rho} = A_{x,r_{b'}} \cup A_{r_{b'},r_{b}} \cup A_{r_{b},\rho},$
it is now immediate that the answer to the query over $A_{x,\rho}$ is
    $D[b',\mathtt{lab}(b'),x,Q]+b.cnt[Q]+b'.cnt[Q],$
where $\mathtt{lab}(b')$ is the labeling of the micro-tree $b'.$
Therefore, our algorithm runs in $\O(1)$ time. 
\end{proof}

We now analyze the space cost of our data structures.

\begin{lemma}\label{lemma:pathCountingQueryBaseCase0Space}
    The data structures in this section occupy $\O(n\lg\lg{n})$ bits when $\eps\in(0,\frac{1}{4d}).$
\end{lemma}  
\begin{proof}
To analyze the space cost, we tally up the costs of the main constituents
of our data structure: the $cnt$-arrays 
stored at the roots of each
mini- and micro-tree, and the $D$-table.

There being $\Theta(n/(c^{2d}\lg{n}))$ mini-trees,
each of which contains an array of $c^{2d}$ elements, $\O(\lg{n})$ bits each,
the associated $cnt$-arrays contribute $\O(n/(c^{2d}\lg{n})\times{}c^{2d}\times{}\lg{n}) = \O(n)$ bits.

Analogously, for the $\Theta(n/c^{2d})$ micro-trees, 
the net contribution of the associated $cnt$-arrays is
$\O(n/c^{2d}\times{}c^{2d}\times{}\lg{(c^{2d}\lg{n})}) = \O(n\lg\lg{n})$ bits,
which is the space claimed in the statement.

It suffices to show, therefore, that the space occupied by the $D$-structure
can not exceed $\O(n\lg\lg{n})$ bits; as demonstrated below, it is much less.
Indeed, as mentioned in \Cref{lemma:lemmaOnTreeCover01},
each micro-tree can have up to $2c^{2d}$ nodes, which gives us at most 
$2^{2\cdot{}2c^{2d}}$ possible topologies $\tau.$
In turn, each of the $2c^{2d}$ nodes can independently be assigned $c^{d}$ possible weight vectors;
hence the number of possible configurations $\lambda$ is at most $(c^{d})^{2c^{2d}}$
    (the number of strings of length $2c^{2d}$ over alphabet $[c^d]$).
Furthermore, the $c^{2d}$ query orthogonal ranges $Q$
make for 
$2c^{2d} \times{} c^{2d} = 2c^{4d}$ distinct queries,
i.e. the number of nodes times the number of ranges.
The number of entries in the table $D$ is thus
at most $\O(2^{4c^{2d}}\cdot{}(c^d)^{2c^{2d}}\cdot{}2c^{4d}) = \O((4c^d)^{2c^{2d}}c^{4d}).$
The term $c^{4d}$ is $\smallO(\lg{n}).$
To upper-bound the term $(4c^d)^{2c^{2d}},$ we notice
    that $c^d = \ceil{\lg^{\eps}{n}}^{d} < \sqrt[4]{\lg{n}}/4$ for sufficiently large $n$. 
Therefore, we have the following chain of inequalities for sufficiently large $n$:
\[
    \begin{array}{lllllll}
        (4c^d)^{2c^{2d}} &<& \Bigl(\sqrt[4]{\lg n}\Bigr)^{2c^{2d}} &<& \Bigl(\sqrt[4]{\lg n}\Bigr)^{2\sqrt{\lg n}} &=& {\Bigl( \sqrt{\lg{n}} \Bigr)^{\sqrt{\lg{n}}}}\\
         &=& {\Bigl( {2^{\lg\sqrt{\lg{n}}}} \Bigr)^{\sqrt{\lg{n}}}} &=& {2^{{\sqrt{\lg{n}}} \cdot \lg\sqrt{\lg{n}}}} &<& {2^{\sqrt{\lg{n}} \cdot \frac{\sqrt{\lg n}}{2} } = \sqrt{n}}
    \end{array}
\]
Thus, the number of entries in $D$ is at most $\O(\sqrt{n}\lg^{4d\eps}{n}).$
Each entry holding a value of $\O(\lg{c^{2d}}) = \O(\lg\lg{n})$ bits,
the table $D$ occupies $\O(\sqrt{n}\lg\lg{n}\lg^{4d\eps}{n}) = \smallO(n)$ bits,
in total. 

Finally, as shown above, the number of ways to assign
weight vectors to nodes of a micro tree is a
$\O(L'\lg\lg{n})$-bit number.
Thus, the storage space for the labelings of each of the $\Theta(n/L')$ micro-trees
amounts to $\O(n\lg\lg{n})$ bits. 
\end{proof}
With \Cref{lemma:pathCountingQueryBaseCase0Time,lemma:pathCountingQueryBaseCase0Space}, we have the following
\begin{lemma}\label{lemma:pathCountingBaseCase0}
    Let $d \geq 0,\,\eps \in (0, \frac{1}{4d})$ be constants.
    A tree $T$ on $n$ nodes, in which each node is assigned a $\dimcl{0}{d}{\eps}$-dimensional weight vector, 
    can be represented in $\O(n\lg\lg{n})$ bits of space such that a path counting
	query is answered in $\O(1)$ time.
\end{lemma}

Finally, instantiating \Cref{section:masterLemmaGeneral} with $g(x) \equiv 1$ and $\oplus$ as the regular addition
operation $\mathtt{+}$ in $\mathbb{R},$
we can apply \Cref{lem:reductionGeneral} to \Cref{lemma:pathCountingBaseCase0} iteratively and obtain the following result:

\begin{theorem}\label{theorem:pathCounting}
  Let $d \geq 1$ be a constant integer.
  A tree $T$ on $n$ nodes, in which each node is assigned a $d$-dimensional weight vector,  
  can be represented in $\O(n(\lg n/\lg\lg n)^{d-1})$ words such that a 
  path counting query can be answered in $\O((\lg{n}/\lg\lg n)^{d})$ time.
\end{theorem}

\section{Path Reporting}\label{appendix:pathReporting}
We use the following result of Chan et al.~\cite{DBLP:journals/algorithmica/ChanHMZ17}:
\begin{lemma}[\cite{DBLP:journals/algorithmica/ChanHMZ17}]\label{lemma:pathReportingBaseCase}
    An ordinal tree on $n$ nodes whose weights are drawn from $[n]$
    can be represented using $\O(n\lg^{\eps}{n})$ words of space,
    such that path reporting queries can be supported in 
    $\O(\lg\lg{n}+k)$ time, where $k$ is the number of reported nodes and $\eps$
    is an arbitrary positive constant.
\end{lemma}

\Cref{lemma:pathReportingBaseCase} implies the following
\begin{lemma}\label{lemma:d1eReporting}
    Let $d \geq 1$ and $0 < \eps < \frac{1}{2(d-1)}$ be constants,
    and let $T$ be an ordinal tree on $n$ nodes,
    in which each node is assigned a $\dimcl{1}{d}{\eps}$-dimensional
    weight vector. Then, $T$ can be represented in
	$\O(n\lg^{\eps'}{n})$ words of space,
    for any $\eps' \in (2(d-1)\eps,1),$ so that
	a path reporting query can be answered in $\O(\lg\lg{n}+k)$ time,
    where $k$ is the number of reported nodes.
\end{lemma}
\begin{proof}
    In brief, we build a path reporting data structure from \Cref{lemma:pathReportingBaseCase}
    for each possible orthogonal range over the last $(d-1)$ dimensions.
    When presented with a query, we directly proceed to the appropriately-tagged (by the last
    $(d-1)$ weights) reporting structure, and launch the query therein. 
    A detailed exposition follows.

    We assume the encoding of $T$ as in \Cref{lemma:smallAlphabetTrees}; the space
    incurred is only $\O(n)$ bits, i.e. negligible with respect
    to the terms derived below.

    For any $\dimcl{0}{d-1}{\eps}$-dimensional orthogonal range $G$ 
    we build an explicit scalarly-weighted tree $E_{G}$ as the
    extraction of the node set $\{v\,|\,v \in T \text{ and } \pnt{w}_{2,d}(v) \in G\}$ from $T$.
    The weight of a node in $E_{G}$ is the $1$st weight of its $T$-source.
    If $E_{G}$ has a dummy root, then its weight is $-\infty.$ 

    $E_{G}$ is represented using \Cref{lemma:pathReportingBaseCase},
    in space that is at most $\O(n\lg^{\delta}{n})$ words, for arbitrarily small positive $\delta.$
    In order to adjust the nodes between $T$ and $E_{G},$ a $0/1$-labeled
    tree $T_{G}$ is maintained. It has the same structure as $T,$
    and a node in it has label $1$ {\textit{iff}} it has been extracted into $E_{G};$
    its label is $0,$ otherwise.
    The tree $T_{G}$ is represented using \Cref{lemma:smallAlphabetTrees},
    in $3n+\smallO(n)$ bits.

    Accounting for all possible ranges $G,$ we therefore have $\O(n\lg^{\delta+2(d-1)\eps}{n})$ words of space, in total.
    Setting $\delta$ to be sufficiently small and assigning $(\delta+2(d-1)\eps)$ to $\eps'$
    justifies the space claimed. All the $T_{G}$-structures collectively
    occupy $\O(n\lg^{2(d-1)\eps}{n})$ bits of space, which is $\smallO(n)$ words.

    Let $P_{x,y}$ and $Q$ be, respectively, the path and the orthogonal range given as the query parameters.
    Using the notation of the \Cref{lem:reductionBinaryCase}, and for the same reasons
    as given therein, we concern ourselves only with answering the query over the path $A_{x,z},$
    where $z = \mathtt{LCA}(T,x,y)$ ($\mathtt{LCA}$ is available through $T$'s encoding).
    To answer the query, we first locate the relevant tree $E_{Q'},$ where $Q' = Q_{2,d},$
    and launch path reporting query in $E_{Q'}$, having adjusted the nodes $x$ and $z$
    accordingly to their $T_{Q'}$-{\astralbody}s
    as $x = \mathtt{pre\_rank}_{1}(T_{Q'},\mathtt{level\_anc}_{1}(T_{Q'},x,1))$,
    and analogously for $z.$
    Finally, the one-dimensional query in $E_{Q'}$ executes in $\O(\lg\lg{n}+k)$ time, 
    by \Cref{lemma:pathReportingBaseCase}, thereby establishing the claimed time bound.
    For each returned node $x$, we recover its original identifier
    with $\mathtt{pre\_select}_1(T_{Q'},x)$.
\end{proof}

Instantiating \Cref{section:masterLemmaGeneral} with $g(x) = \{x\}$ and the semigroup sum operator $\oplus$ as the set-theoretic
union operator $\cup$, 
\Cref{lem:reductionGeneral} and \Cref{lemma:d1eReporting} combined imply the following
\begin{theorem}\label{theorem:pathReporting}
  Let $d \geq 2$ be a constant integer.
  A tree $T$ on $n$ nodes, in which each node is assigned a $d$-dimensional weight vector,  
  can be represented in $\O(n\lg^{d-1+\eps}{n})$ 
  words such that a path reporting query can be answered in $\O((\lg^{d-1}{n})/(\lg\lg{n})^{d-2}+k)$
  time where $k$ is the number of the nodes reported, for an arbitrarily small positive constant $\eps$.
\end{theorem}

\bibliography{bibliography}
\end{document}